\documentclass[12pt]{article}
\usepackage[top=25truemm,bottom=25truemm,left=20truemm,right=20truemm]{geometry}
\title{\bf General condition of quantum teleportation \\ by one-dimensional quantum walks}
\author{Tomoki Yamagami \thanks{Department of Information Physics and Computing, Graduate School of Information Science and Technology, The University of Tokyo, Tokyo, 113-8656, Japan.}
\and Etsuo Segawa \thanks{Graduate School of Environment and Information Sciences, Yokohama National University
Yokohama, 240-8501, Japan.}
\and Norio Konno \thanks{Department of Applied Mathematics, Faculty of Engineering, Yokohama National University, Yokohama, 240-8501, Japan.}
}

\usepackage{fancybox}
\usepackage{latexsym,amssymb,amsfonts,amsmath,amsthm}
\usepackage{ascmac}
\usepackage{braket}
\usepackage[dvips]{graphicx}
\usepackage{wrapfig}
\usepackage{enumerate}
\usepackage{comment}
\usepackage{empheq}
\usepackage{cite}
\newcommand{\vr}
{v_R}
\newcommand{\vl}
{v_L}
\newcommand{\T}
{^{\rm T}}
\newcommand{\twobytwo}[4]
{\left[ \begin{array}{cc}
#1 & #2 \\
#3 & #4
\end{array} \right]}
\newcommand{\threebythree}[9]
{\left[ \begin{array}{ccc}
#1 & #2 & #3 \\
#4 & #5 & #6 \\
#7 & #8 & #9
\end{array} \right]}
\newcommand{\onebytwo}[2]
{\left[ \begin{array}{c} #1 \\ #2 \end{array} \right]}
\newcommand{\twovec}[2]
{\left[ \begin{array}{cc}
#1  \\
#2 
\end{array} \right]}
\newcommand{\Hp}
{\mathcal{H_{\rm P}}}
\newcommand{\Hcn}
{\mathcal{H_{\rm C}}}
\newcommand{\Hc}[1]
{\mathcal{H_{\rm C}^{\rm (#1)}}}
\newcommand{\agmnt}
{{\rm arg}}
\newcommand{\bfit}[1]
{\mbox{\boldmath $#1$}}
\usepackage{comment}
\newcommand{\rrule}
{\vrule width 0pt
height13pt depth 8pt}
\newcommand{\rrrule}
{\vrule width 0pt
height20pt depth 15pt}

\newtheorem{proposition}{Proposition}
\newtheorem{definition}[proposition]{Definition}
\newtheorem{theorem}[proposition]{Theorem}
\newtheorem{remark}[proposition]{Remark}
\newtheorem{lemma}[proposition]{Lemma}
\newtheorem{corollary}[proposition]{Corollary}

\begin{document}
\maketitle

\noindent{\bf Abstract}\quad
We extend the scheme of quantum teleportation by quantum walks introduced by Wang et al. (2017). First, we introduce the mathematical definition of the accomplishment of quantum teleportation by this extended scheme. Secondly, we show a useful necessary and sufficient condition that the quantum teleportation is accomplished rigorously.  Our result classifies the parameters of the setting for {the accomplishment of quantum teleportation}.

\section{Introduction}\label{intro}
Quantum walk is considered as a quantum analogue of random walk. This model was first introduced in the context of quantum information theory such as Aharonov et al. \cite{ADZ93} and Ambainis et al. \cite{ABNV01}. Since then, quantum walk is treated as an interesting model in the field of mathematics and information theory \cite{CFG02, K03, Kn02, VA08, VA12} and expected of its application \cite{P13, KI19}. Quantum walk is capable of universal quantum computation and able to be implemented by the physical system in various ways \cite{C09, LCETK10, CGW13, KRBD10}, which is why the model is considered to be expectable one.\par
On the other hand, quantum teleportation is a communication protocol that transmits a quantum state from one place to another. It is first introduced by Bennett et al. \cite{BBCJPW93} and regarded as not only a system for communication but also the basis of quantum computation \cite{TMFLF13}.\par
Recently, the works on applications of quantum walks to quantum teleportation {\cite{WSX17, SWLR19, LCWHL19, ZYLZ20} appear}. In previous quantum teleportation systems, they had to produce prior entangled states and carried on transmission with it. However, by using quantum walks, the walk itself has a role of entanglement, which makes teleportation simpler. In the previous study \cite{WSX17}, the concrete models of teleportation by quantum walks are shown, but the general condition where the scheme of teleportation succeeds is not shown. In this paper, we extend the scheme of quantum teleportation by quantum walks introduced by Wang et al. \cite{WSX17}. We introduce the mathematical definition of the accomplishment of quantum teleportation by this extended scheme. Then, we show a useful necessary and sufficient condition for it.  Our result classifies the parameters of the setting for the accomplishment of the quantum teleportation including Wang et al.'s settings.\par
 The rest of the paper is organized as follows. Section 2 gives the definition of our quantum walk model, and in Sect. 3 we give the scheme of teleportation by the quantum walk model. In Sect. 4, we present our main theorem of this paper and demonstrate some examples of the theorem.  Furthermore, Sect. 5 is devoted to the proof of the result. Finally, we give {a} summary and discussion in Sect. 6. 

\section{Quantum Walks}\label{sec:1}
Here, we introduce the quantum walks (QWs). First, we review a basic model of discrete QW and then introduce the QW applied to the scheme of quantum teleportation.

\subsection{The One-Coin Quantum Walks on One-Dimensional Lattice}\label{sec:2}
The one-dimensional quantum walk with one coin is defined in a compound Hilbert space of the position Hilbert space $\mathcal{H}_{\rm P} = {\rm span}\{ \ket{x} | x \in \mathbb{Z} \}$ and the coin Hilbert space $\mathcal{H}_{\rm C} = {\rm span}\{ \ket{R},\,\ket{L} \}$ with
\begin{eqnarray}
\ket{R} = \onebytwo{1}{0},\quad \ket{L} = \onebytwo{0}{1}. \nonumber
\end{eqnarray}
Note that $\Hcn$ is equivalent to $\mathbb{C}^2$. Then, the whole system is described by $\mathcal{H} = \mathcal{H}_{\rm P} \otimes \mathcal{H}_{\rm C}$. \par
Now, we define one-step time evolution of the quantum walk as $W = \hat{S} \cdot \hat{C}$, where $\hat{S}$ is a shift operator described by
\begin{eqnarray}
\hat{S} = S \otimes \ket{R}\bra{R} + S^{-1} \otimes \ket{L}\bra{L}\nonumber
\end{eqnarray}
with
\begin{eqnarray}
S = \sum_{x \in \mathbb{Z}} \ket{x+1}\bra{x},\nonumber
\end{eqnarray}
and $\hat{C}$ is a coin operator defined by
\begin{eqnarray}
\hat{C} = I_2 \otimes C,\nonumber
\end{eqnarray}
with
\begin{eqnarray}
I_2 = \twobytwo{\,1\,}{\,0\,}{\,0\,}{\,1\,},\quad C \in {\rm U}(2).\nonumber
\end{eqnarray}
Here, U($n$) is the set of $n\times n$ unitary matrices. 

\subsection{$m$-Coin Quantum Walks on One-Dimensional Lattice}\label{sec:MCQW}
To implement schemes of quantum teleportation based on quantum walks, we need to define quantum walks with many coins, which are determined on the whole system $\mathcal{H} = \Hp \otimes \Hcn^{\otimes m}$ with $m\geq n$ (the previous case was one coin QW).\par
Now, we define one-step time evolution of the $m$-coin quantum walk at time $n$ as $W_n = \hat{S}_n \cdot \hat{C}_n$, where $\hat{S}_n$ is a shift operator described by
\begin{eqnarray}\nonumber
\hat{S}_n &=& S \otimes \left(I_{2} \otimes \cdots \otimes I_{2} \otimes \overbrace{\ket{R}\bra{R}}^{n} \otimes I_{2} \otimes \cdots \otimes I_{2}\right) \nonumber \\
&& + S^{-1} \otimes \left(I_{2} \otimes \cdots \otimes I_{2} \otimes \overbrace{\ket{L}\bra{L}}^{n} \otimes I_{2} \otimes \cdots \otimes I_{2}\right),\nonumber
\end{eqnarray}
and $\hat{C}_n$ is the coin operator described by
\begin{eqnarray}
\hat{C}_n = I_{\infty} \otimes \left(I_{2} \otimes \cdots \otimes I_{2} \otimes \overbrace{C_n}^n \otimes I_{2} \otimes \cdots \otimes I_{2}\right).\nonumber
\end{eqnarray}
Here, ``$\overbrace{}^n$" means that the matrix corresponds to $n$th $\Hcn$ and $C_n \in {\rm U}(2)$.\par
Moreover, we put 
\begin{eqnarray}
P_n = \ket{L}\bra{L}C_n,\quad Q_n = \ket{R}\bra{R}C_n.\nonumber
\end{eqnarray}
We should note that $C_n = P_n + Q_n$.  Then, a quantum walker at time $n$ moves one unit to the left with the weight 
\begin{eqnarray}
I_{2} \otimes \cdots \otimes I_{2} \otimes \overbrace{P_n}^n \otimes I_{2} \otimes \cdots \otimes I_{2},\nonumber
\end{eqnarray}
or to the right with weight 
\begin{eqnarray}
I_{2} \otimes \cdots \otimes I_{2} \otimes \overbrace{Q_n}^n \otimes I_{2} \otimes \cdots \otimes I_{2}.\nonumber
\end{eqnarray}
In other words, for $n\in \mathbb{Z}_{\geq}$ and $\ket{\varPsi_n}$, the state of the system at time $n$, the relationship between the states $\ket{\varPsi_n}$ and $\ket{\varPsi_{n+1}}$ is described as
\begin{eqnarray}
\ket{\varPsi_{n+1}} = W_{n+1}\ket{\varPsi_n}.\nonumber
\end{eqnarray}

\section{Schemes of Teleportation}
Let us set $\Hp\otimes \Hc{\rm A}$ and $\Hc{\rm B}$ as the Alice and Bob's spaces, respectively after the fashion of the proposed idea by \cite{WSX17}. Here, $\Hc{\rm A}$, $\Hc{\rm B}\cong \mathbb{C}^2$. In this section, we consider quantum teleportation described in Figure 1. Now, the sender Alice wants to send $\ket{\phi} \in \Hc{\rm A} (\cong  \mathbb{C}^2)$ with $\|\phi\| = 1$ to the receiver Bob. We call $\ket{\phi}$ the target state.  \par
The space of this quantum teleportation is denoted by $\mathcal{H} = \Hp \otimes \Hc{A}\otimes \Hc{B}$. We set the initial state as
\begin{eqnarray}
\ket{\varPsi_0} = \ket{0} \otimes \ket{\phi} \otimes \ket{\psi} \in \mathcal{H}.\nonumber
\end{eqnarray}
Here, $\ket{\psi}$ satisfies $\|\psi\| =1$.  In the framework of  quantum walk, the total state space of quantum teleportation is isomorphic to a two-coin quantum walk whose position Hilbert space is $\Hp$ and whose coin Hilbert space is $\Hc{A}\otimes \Hc{B}$. On the other hand, from the point of view of quantum teleportation, Alice has two initial states $\ket{0} \otimes \ket{\phi} \in \Hp \otimes \Hc{A}$ and Bob has an initial state $\ket{\psi} \in \Hc{B}$, and the goal of the teleportation is that Bob obtains the state $\ket{\phi}$ as the element of $\Hc{B}$.\par
\begin{figure}[b]
\centering \includegraphics[width=140mm]{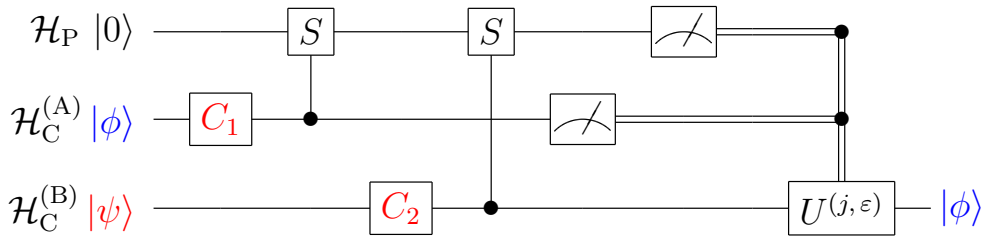}
\caption{Circuit diagram of quantum teleportation by 2-coin quantum walks}
\label{fig:1}       % Give a unique label
\end{figure}
Then, we provide three stages: (1) time evolution, (2) measurement and (3) transformation. 

\subsection{Time Evolution by QW}
In the first stage, we take 2 steps of QWs with two coins; we describe the time evolution operator at the first and second steps $W_1$,\,$W_2$ as
\begin{eqnarray}
W_1 = \hat{S_1} \cdot \hat{C_1} = (S \otimes \ket{R}\bra{R} \otimes I_{2} + S^{-1} \otimes \ket{L}\bra{L} \otimes I_{2})(I_{\infty} \otimes C_1 \otimes I_{2}), \nonumber \\
W_2 = \hat{S_2} \cdot \hat{C_2} = (S \otimes I_{2} \otimes \ket{R}\bra{R} + S^{-1} \otimes I_{2} \otimes \ket{L}\bra{L})(I_{\infty} \otimes I_{2} \otimes C_2), \nonumber
\end{eqnarray}
respectively. Suppose $\ket{\varPsi_n}\in \mathcal{H}$ $(n=0,1,2)$ is the state after the $n$-th time evolution of the QW, and we regard the initial state of $\ket{\varPsi_0}$ of the quantum teleportation as the initial state of the QW. We run this QW for two steps, that is,
\begin{eqnarray}
\ket{\varPsi_0}\stackrel{W_1}{\longmapsto} \ket{\varPsi_1}\stackrel{W_2}{\longmapsto}\ket{\varPsi_2}.\nonumber
\end{eqnarray}
\subsection{Measurement}
In the second stage, to carry out the measurement on the Alice's state, we introduce the observables denoted by self-adjoint operators $M_1$ and $M_2$ on $\Hc{A}$ and $\Hp$, respectively, as follows:
\begin{eqnarray}
M_1 &=& (+1)\ket{\eta_R}\bra{\eta_R} + (-1) \ket{\eta_L}\bra{\eta_L},\nonumber \\
M_2 &=& \sum_{j \in \mathbb{Z}} \frac{\,j\,}{2} \ket{\xi_j}\bra{\xi_j},\nonumber
\end{eqnarray}
where $\ket{\eta_\varepsilon} = H_1\ket{\varepsilon}\,(\varepsilon \in \{R,\,L\})$, and $\ket{\xi_j} = H_2\ket{j}\,(j \in \mathbb{Z})$. Here, $H_1$ and $H_2$ are unitary operators on $\Hc{A}(\cong\mathbb{C}^2)$ and $\Hp(\cong\ell^2({\mathbb{Z}}))$, respectively. Especially, $H_2$ is described as follow:
\begin{eqnarray}\nonumber
&H_2 \simeq 
\left[ \begin{array}{ccc|ccc}
\alpha_{22} & \alpha_{20} & \alpha_{2{(-2)}} &&&\\
\alpha_{02} & \alpha_{00} & \alpha_{0{(-2)}} &\multicolumn{2}{c}{\raisebox{-3pt}[0pt][0pt]{\LARGE $\,\,\,O$}}&\\ 
\alpha_{{(-2)}2} & \alpha_{{(-2)}0} & \alpha_{{(-2)}{(-2)}} &&&\\ \hline
&&&&&\\
\multicolumn{2}{c}{\raisebox{3pt}[0pt][0pt]{\LARGE $\,\,\,\,\,\,O$}}&&\multicolumn{2}{c}{\raisebox{3pt}[0pt][0pt]{\LARGE $\,\,\,I$}}& 
\end{array} \right] = 
\left[ \begin{array}{ccc|ccc}
&&&&&\\
\multicolumn{2}{c}{\raisebox{3pt}[0pt][0pt]{\LARGE $\,\,\tilde{H}_2$}}&&\multicolumn{2}{c}{\raisebox{3pt}[0pt][0pt]{\LARGE $\,\,O$}}&\\ \hline
&&&&&\\
\multicolumn{2}{c}{\raisebox{3pt}[0pt][0pt]{\LARGE $\,\,O$}}&&\multicolumn{2}{c}{\raisebox{3pt}[0pt][0pt]{\LARGE $\,\,I$}}& 
\end{array} \right],&
\end{eqnarray}
where
\begin{eqnarray}\nonumber
\tilde{H}_2 = 
\left[ \begin{array}{ccc}
\alpha_{22} & \alpha_{20} & \alpha_{2{(-2)}} \\
\alpha_{02} & \alpha_{00} & \alpha_{0{(-2)}} \\
\alpha_{{(-2)}2} & \alpha_{{(-2)}0} & \alpha_{{(-2)}{(-2)}} 
\end{array} \right].
\end{eqnarray}
The computational basis of $H_2$ in RHS is $\{\ket{2},\,\ket{0},\,\ket{-2},\,\hdots\}$ by this order. 
The observed values of the observable $M_1$ are $\varepsilon\in \{\pm 1\}$ after the description of \cite{WSX17}, but in this paper, we describe the observed values of $M_1$ by $R, L$ by the bijection map 
\[ R\leftrightarrow +1 \text{ and } L\leftrightarrow -1.\] 
In the same way, we describe the observed values of $M_2$ as $\{-2,0,2\}$ by the bijection map 
\[ 2k \leftrightarrow k \;(k=-1,0,1). \] 
Furthermore,  we extend the domains of operators $M_1$ and $M_2$  to the whole system $\mathcal{H}$ by putting $M_1^{\rm (s)}$ and $M_2^{\rm (s)}$ as follows:
\begin{align*}
M_1^{\rm (s)} &:= I_{\infty} \otimes 
%(1\ket{\eta_R}\bra{\eta_R} + (-1) \ket{\eta_L}\bra{\eta_L} )}
M_1 \otimes I_{\Hc{B}}, \nonumber \\
M_2^{\rm (s)} &:= 
%\left( \sum_{j \in \mathbb{Z}} \frac{\,j\,}{2} \ket{\xi_j}\bra{\xi_j} \right)
M_2 \otimes I_{\Hc{A}} \otimes I_{\Hc{B}}. \nonumber
\end{align*}
This means that Alice carries out projection measurements on $\Hc{A}$ and $\Hp$ with the eigenvectors $\mathcal{B}_1 = \{ \ket{\eta_\varepsilon} | \varepsilon \in \{R,\,L\} \}$ of $M_1$ and $\mathcal{B}_2 = \{ \ket{\xi_j} | j \in \mathbb{Z} \}$ of $M_2$, respectively. 
If Alice gets the observed values $\varepsilon$ by $M_1$ and $j$ by $M_2$, respectively, then the states collapse to $\ket{\eta_\varepsilon}\in\mathcal{H}_C^{(A)}$ and $\ket{\xi_j}\in \mathcal{H}_P$, respectively.
%\red{In the following, we regard ...}\par

Through the measurements,  if the state of $\Hc{A}$ collapses to $\ket{\eta_\varepsilon} \in \mathcal{B}_1$ by $M_1$ and the state of $\Hp$ collapses to $\ket{\xi_j} \in \mathcal{B}_2$ by $M_2$, the degenerate state on the whole state is denoted by $\ket{\varPsi_*^{(j,\,\varepsilon)}}\in \mathcal{H}$. So, the state $\ket{\varPsi_*^{(j,\,\varepsilon)}}$ can be described explicitly as follows. The proof is given in Sect. 5.
\begin{proposition}\label{prop:finalstate}
{\rm The state $\ket{\varPsi_*^{(j,\,\varepsilon)}}$ can be described as}
\begin{eqnarray}
\ket{\varPsi_*^{(j,\,\varepsilon)}} = \ket{\xi_j} \otimes \ket{\eta_\varepsilon} \otimes \ket{\varPhi_*^{(j,\,\varepsilon)}}, \label{state}
\end{eqnarray}
{\rm where $\ket{\varPhi_*^{(j,\,\varepsilon)}} = V^{(j,\,\varepsilon)}\ket{\phi}$ and $V^{(j,\,\varepsilon)}$ is a linear map on $\Hc{B}$ (See (\ref{vje}) for the detailed expression for $V^{(j,\,\varepsilon)}$).}
\end{proposition}\par
Then, our problem is converted to finding a practical necessary and sufficient condition for the unitarity of  $V^{(j,\,\varepsilon)}$.
%%%%%%%%%%
\subsection{Transformation}
In the final stage, Bob should convert his state $\ket{\varPhi_*^{(j,\,\varepsilon)}}\in \Hc{B}$ to the state $\ket{\phi}$. After the measurements, Alice sends the outcomes $\varepsilon\in\{L,R\}$ and $j\in\{-2,0,2\}$ to Bob. Then, Bob acts a unitary operator $U^{(j,\,\varepsilon)}$ on $\Hc{B}$ to $\ket{\varPhi_*^{(j,\,\varepsilon)}}$, depending on a pair of  observed results $(j,\,\varepsilon)$. Finally, Bob obtains a state $\ket{\varPhi}:=U^{(j,\,\varepsilon)}\ket{\varPhi_*^{(j,\,\varepsilon)}} \in \Hc{B}$. If $\ket{\varPhi}=\ket{\phi}$, we can regard that the teleportation is ``accomplished" (we define this clearly below).
\subsection{A mathematical formulation of schemes of teleportation}
In the above subsections, we introduced the notion of quantum teleportation driven by quantum walk. As we have seen, the factors to determine the scheme of this teleportation are Bob's initial state $\ket{\psi}$, the coin operators $C_1$ and $C_2$, and the measurement operator $H_1$ and $H_2$. Then, for convenience, we define the set of them as the parameter of the teleportation as follows:
\begin{definition}
{\rm We call 
\begin{align}\nonumber
\bfit{T}= (\ket{\psi};\,C_1,\,C_2;\,H_1,\,H_2) \in {\rm \mathbb{C}^2 \times U(2) \times U(2) \times U(2) \times U(\infty)}
\end{align} a {\bf quantum walk measurement procedure}.}
\end{definition}
\begin{definition}
{\rm Let $\ket{\varPhi}\in \Hc{B}$ be a Bob's final state of a quantum walk measurement procedure ${\bfit T}$ and $\ket{\phi}\in \Hc{A}$ be the target state.
If this quantum walk measurement procedure ${\bfit T}$ satisfies $\ket{\varPhi} = \ket{\phi}$ for any observed value $(j,\,\varepsilon)\in \{ -2,0,2 \}\times \{L,R\}$ by Alice, we say that {\bf the quantum teleportation is accomplished by {\bfit T}}. }
\end{definition}
\begin{definition}
{\rm We define $\mathcal{T} \subset {\rm \mathbb{C}^2 \times U(2) \times U(2) \times U(2) \times U(\infty)}$ by
\begin{eqnarray}\nonumber
\mathcal{T} := \left\{ \bfit{T} = (\ket{\psi};\,C_1,\,C_2;\,H_1,\,H_2)\,|\,\scalebox{0.8}[1]{\rm {\bfit T}\, accomplishes\, the\, quantum\, teleportation.} \right\}
\end{eqnarray}
and call $\mathcal{T}$ {\bf the class of quantum teleportation driven by 2-coin quantum walks}.}
\end{definition}\par
The main purpose of this paper is to determine explicitly the class $\mathcal{T}$.
\section{Our result} 
In this section, we present our main result on the quantum teleportation by quantum walks.
\subsection{Main Theorem}
\begin{theorem}\label{thm:main}
{\rm Quantum walk measurement procedure $\bfit{T} = (\ket{\psi};\,C_1,\,C_2;\,H_1,\,H_2)$ accomplishes the quantum teleportation, i.e., $\bfit{T} \in \mathcal{T}$ iff $\bfit{T}$ satisfies the following three conditions simultaneously}:
\begin{enumerate}[{\bf (I)}]
\item {\bf [Condition for $H_1$]} {\rm $|\braket{R|H_1|R}|=|{\braket{R|H_1|L}}|$.}
%%%%%
\item {\bf [Condition for $C_2$ and $\psi$\,]} {\rm $\left|\left\langle R|C_{2}| \psi\right\rangle\right|=\left|\left\langle L|C_{2}| \psi\right\rangle\right| = \displaystyle\frac{1}{\sqrt{2}}$.}
%%%%%
\item {\bf [Condition for $H_2$]} {\rm $\bfit{T}$ satisfies one of the following {two} conditions at least:}
\begin{enumerate}
\renewcommand{\labelenumii}{\bf (\roman{enumii})}
%%%%%
\item 
{\rm Let $\bfit{H}$ be the set of three-dimensional unitary matrices defined by 
    \begin{equation}\nonumber \bfit{H}=\left\{  
    \begin{bmatrix} p & r & 0 \\  0 & 0 & t  \\q & s & 0 \end{bmatrix},\;
    \begin{bmatrix} p & 0 & r \\ 0 & t & 0 \\ q & 0 & s  \end{bmatrix},\;
    \begin{bmatrix} 0 & p & r \\ t & 0 & 0 \\ 0 & q & s \end{bmatrix}
    \in {\rm U}(3) \;:\; 
    |p|=|q|   \right\} 
    \end{equation}
Then, $H_2=\tilde{H}_2\oplus I_\infty$ with $\tilde{H}_2\in \bfit{H}$. }
\item {\rm for all $k\in \{0,\,\pm 2\}$,}
\begin{align}\nonumber
|(H_2)_{2k}| = |(H_2)_{{(-2)}k}|
\end{align}
{\rm and}
\begin{align}\nonumber
\agmnt(H_2)_{2k} + \agmnt(H_2)_{{(-2)}k} - 2\agmnt(H_2)_{0k} \in (2\mathbb{Z}+1)\pi.
\end{align}
\end{enumerate}
%\red{!!(ii)$\Rightarrow$ (iii) !!}\\
{\rm Here, $(H_2)_{jk} =\braket{j|H_2|k}$.}
\end{enumerate}
{\rm Moreover, in any case, the transformation $U^{(j,\,\varepsilon)}$ by Bob depending on observed results $(j,\,\varepsilon)$ is unitary described as}
\begin{align}\nonumber
U^{(j,\,\varepsilon)} = \frac{1}{\|V^{(j,\,\varepsilon)}\ket{\phi}\|} \left(V^{(j,\,\varepsilon)}\right)^{-1},
\end{align}
{\rm where}
%\begin{align}\nonumber
%V^{(j,\,\varepsilon)}=\twovec {\bra{\eta_\varepsilon}(\overline{\alpha_{2k}}\ket{R}\bra{R}+\overline{\alpha_{0k}}\ket{L}\bra{L})\beta_{R}}{\bra{\eta_\varepsilon}(\overline{\alpha_{0k}}\ket{R}\bra{R} + \overline{\alpha_{(-2)k}}\ket{L}\bra{L})\beta_{L}}C_1,
%\end{align}
\begin{align}\nonumber
V^{(j,\,\varepsilon)}=\twovec {\bra{\eta_\varepsilon}(\overline{\alpha_{2j}}Q_1+\overline{\alpha_{0j}}P_1)\beta_{R}}{\bra{\eta_\varepsilon}(\overline{\alpha_{0j}}Q_1 + \overline{\alpha_{(-2)j}}P_1)\beta_{L}},
\end{align}
{\rm regardless of $\ket{\phi}$. Here $\alpha_{jk}=(H_2)_{jk}$ and $\beta_L=\bra{L}C_2\ket{\psi}$, $\beta_R=\bra{R}C_2\ket{\psi}$.}
\end{theorem}

\begin{remark}
{\rm This theorem implies that accomplishment of the quantum teleportation is independent of $C_1$. Moreover, the theorem does not depend on $C_2$ and $\ket{\psi}$, for each one, but 
``$C_2\ket{\psi}$." After all,  the accomplishment of quantum teleportation is determined only by three factors, that is, $H_1$, $H_2$, and $\ket{\psi'}=C_2\ket{\psi}$; this is a generalization of the statement of \cite{WSX17}. }
\end{remark}

\begin{remark}
{\rm The condition {\bf (II)} means that the coin operator $C_2$ must be unbiased. This claim agrees with Li et al. \cite{LCWHL19}, in which it is the case of the number of qubit $N=1$.}
\end{remark}

\subsection{Examples and Demonstrations}
In the following, we put $H=\displaystyle\frac{1}{\sqrt{2}} \twobytwo{1}{1}{1}{-1}$.
\begin{enumerate}[(1)]
\item\quad We choose
\begin {eqnarray}
\ket{\psi} = \ket{R},\,\,C_1 = I_2,\,\,C_2 = H_1=H,\,\,\tilde{H}_2 \simeq  H\oplus {I_{\infty}}. \nonumber
\end{eqnarray}
This case satisfies {\bf (III)-(i)} and Wang et al.\cite{WSX17} has shown that in this case the quantum teleportation is accomplished. Bob's state before measurement $\ket{\varPhi^{(j,\,\varepsilon)}}$ and the operator $U^{(j,\,\varepsilon)}$ are as follows:
\begin{align*}
\begin{array}{|c|c|c|}
\hline
(j,\,\varepsilon) &\quad\quad\ket{\varPhi^{(j,\,\varepsilon)}}\quad\quad &\quad\quad U^{(j,\,\varepsilon)}\quad\quad \rrule\\ \hline \hline
(2,\,R) & \ket{\phi} & I_2 \rrule\\ \hline
(0,\,R) & X\ket{\phi} & X \rrule\\ \hline
(-2,\,R) & Z\ket{\phi} & Z \rrule\\ \hline
(2,\,L) & Z\ket{\phi} &  Z \rrule\\ \hline
(0,\,L) & XZ\ket{\phi} & ZX \rrule\\ \hline
(-2,\,L) & \ket{\phi} & I_2 \rrule\\ \hline
\end{array}
\end{align*}

\item We choose
\begin {eqnarray}
\ket{\psi} = \frac{\ket{R}+\ket{L}}{\sqrt{2}},\,\,C_1 = C_2 = I_2,\,H_1=H,\,\,
\tilde{H}_2 =  \displaystyle\frac{1}{\sqrt{3}}\threebythree{-e^{ \frac{4}{3}\pi i }}{-1}{-e^{ \frac{2}{3}\pi i }}{1}{1}{1}{e^{ \frac{2}{3}\pi i }}{1}{e^{ \frac{4}{3}\pi i }}. \nonumber
\end{eqnarray}
This case satisfies {\bf (III)-(ii)}. Bob's state before measurement $\ket{\varPhi^{(j,\,\varepsilon)}}$ and the operator $U^{(j,\,\varepsilon)}$ are as follows:
\begin{align*}
\begin{array}{|c|c|c|}
\hline
(j,\,\varepsilon) & \ket{\varPhi^{(j,\,\varepsilon)}} & U^{(j,\,\varepsilon)} \rrule \\ \hline \hline
(2,\,R) &\quad \displaystyle\frac{1}{\sqrt{2}}\twobytwo{e^{\frac{2}{3}\pi i}}{1}{1}{-e^{\frac{4}{3}\pi i}}\ket{\phi} \quad&\quad \displaystyle\frac{1}{\sqrt{2}}\twobytwo{e^{\frac{4}{3}\pi i}}{1}{1}{-e^{\frac{2}{3}\pi i}} \quad \rrrule \\ \hline
(0,\,R) & \displaystyle\frac{1}{\sqrt{2}}\twobytwo{1}{1}{1}{-1}\ket{\phi} & \displaystyle\frac{1}{\sqrt{2}}\twobytwo{1}{1}{1}{-1}  \rrrule \\ \hline
(-2,\,R) & \displaystyle\frac{1}{\sqrt{2}}\twobytwo{e^{\frac{4}{3}\pi i}}{1}{1}{-e^{\frac{2}{3}\pi i}}\ket{\phi} & \displaystyle\frac{1}{\sqrt{2}}\twobytwo{e^{\frac{2}{3}\pi i}}{1}{1}{-e^{\frac{4}{3}\pi i}}  \rrrule \\ \hline
(2,\,L) & \displaystyle\frac{1}{\sqrt{2}}\twobytwo{e^{\frac{2}{3}\pi i}}{-1}{1}{e^{\frac{4}{3}\pi i}}\ket{\phi} & \displaystyle\frac{1}{\sqrt{2}}\twobytwo{e^{\frac{4}{3}\pi i}}{1}{-1}{e^{\frac{2}{3}\pi i}}  \rrrule \\ \hline
(0,\,L) & \displaystyle\frac{1}{\sqrt{2}}\twobytwo{1}{-1}{1}{1}\ket{\phi} & \displaystyle\frac{1}{\sqrt{2}}\twobytwo{1}{1}{-1}{1}  \rrrule \\ \hline
(-2,\,L) & \displaystyle\frac{1}{\sqrt{2}}\twobytwo{e^{\frac{4}{3}\pi i}}{-1}{1}{e^{\frac{2}{3}\pi i}}\ket{\phi} & \displaystyle\frac{1}{\sqrt{2}}\twobytwo{e^{\frac{2}{3}\pi i}}{1}{-1}{e^{\frac{4}{3}\pi i}}  \rrrule \\ \hline
\end{array}
\end{align*}

\item\quad We choose
\begin {eqnarray}
\ket{\psi} = \frac{\ket{R}+i\ket{L}}{\sqrt{2}},\,\,C_1 = C_2 =I_2,\,\,H_1=H,\,\,
\tilde{H}_2 = \threebythree{i/2}{1/\sqrt{2}}{-i/2}{1/\sqrt{2}}{0}{1/\sqrt{2}}{i/2}{-1/\sqrt{2}}{-i/2}. \nonumber
\end{eqnarray}
This case is another example of {\bf (III)-(ii)}. Bob's state before measurement $\ket{\varPhi^{(j,\,\varepsilon)}}$ and the operator $U^{(j,\,\varepsilon)}$ are as follows:
\begin{align*}
\begin{array}{|c|c|c|}
\hline
(j,\,\varepsilon) & \ket{\varPhi^{(j,\,\varepsilon)}} & U^{(j,\,\varepsilon)} \rrule\\ \hline \hline
(2,\,R) & \quad\displaystyle\frac{1}{\sqrt{3}}\twobytwo{i}{\sqrt{2}}{\sqrt{2}i}{-1}\ket{\phi} \quad &\quad \displaystyle\frac{1}{\sqrt{3}}\twobytwo{-i}{-\sqrt{2}i}{\sqrt{2}}{-1}  \quad\rrrule \\ \hline
(0,\,R) & \twobytwo{-1}{0}{0}{i}\ket{\phi} & \twobytwo{-1}{0}{0}{-i}  \rrrule \\ \hline
(-2,\,R) & \displaystyle\frac{1}{\sqrt{3}}\twobytwo{-i}{\sqrt{2}}{\sqrt{2}i}{1}\ket{\phi} & \displaystyle\frac{1}{\sqrt{3}}\twobytwo{i}{-\sqrt{2}i}{\sqrt{2}}{1}  \rrrule \\ \hline
(2,\,L) & \displaystyle\frac{1}{\sqrt{3}}\twobytwo{i}{-\sqrt{2}}{\sqrt{2}i}{1}\ket{\phi} & \displaystyle\frac{1}{\sqrt{3}}\twobytwo{-i}{-\sqrt{2}i}{-\sqrt{2}}{1}  \rrrule \\ \hline
(0,\,L) & \twobytwo{-1}{0}{0}{-i}\ket{\phi} & \twobytwo{-1}{0}{0}{i}  \rrrule \\ \hline
(-2,\,L) & \displaystyle\frac{1}{\sqrt{3}}\twobytwo{-i}{-\sqrt{2}}{\sqrt{2}i}{-1}\ket{\phi} & \displaystyle\frac{1}{\sqrt{3}}\twobytwo{i}{-\sqrt{2}i}{-\sqrt{2}}{-1}  \rrrule \\ \hline
\end{array}
\end{align*}
\end{enumerate}

\section{Proof of Main Theorem}
\subsection{Proof of Proposition~1}
\begin{proof}
At $n=1$, $\ket{\varPsi_0}$ evolves to
\begin{eqnarray}
\ket{\varPsi_1} = W_1\ket{\varPsi_0} = \ket{1} \otimes \ket{Q_1 \phi} \otimes \ket{\psi} + \ket{-1} \otimes \ket{P_1 \phi} \otimes \ket{\psi}, \nonumber
\end{eqnarray}
and at $n=2$, $\ket{\varPsi_1}$ evolves to
\begin{align}
\ket{\varPsi_2} = W_2\ket{\varPsi_1} =& \ket{2} \otimes \ket{Q_1 \phi} \otimes \ket{Q_2 \psi}  \nonumber \\
& + \ket{0} \otimes \left(\ket{Q_1 \phi} \otimes \ket{P_2 \psi} + \ket{P_1 \phi} \otimes \ket{Q_2 \psi}\right) \nonumber \\
& + \ket{-2} \otimes \ket{P_1 \phi} \otimes \ket{P_2 \psi}. \nonumber
\end{align}
If the coin state of Alice collapses to $\ket{\eta_\varepsilon} \in \mathcal{B}_1$ after the observable $M_1$, the total state $\ket{\varPsi_2}$ is changed to
\begin{align}
\ket{\varPsi_*^{(\varepsilon)}} =\frac{1}{\kappa^{(\varepsilon)}}\{& \ket{2} \otimes \ket{\eta_\varepsilon} \otimes \braket{\eta_\varepsilon|Q_1 \phi}\ket{Q_2 \psi} \nonumber \\
&+ \ket{0} \otimes \ket{\eta_\varepsilon} \otimes \left( \braket{\eta_\varepsilon|Q_1 \phi}\ket{P_2 \psi}+ \braket{\eta_\varepsilon|P_1 \phi}\ket{Q_2 \psi}\right) \nonumber \\
&+ \ket{-2} \otimes \ket{\eta_\varepsilon} \otimes \braket{\eta_\varepsilon|P_1 \phi}\ket{P_2 \psi}\}.\nonumber 
\end{align}
Here, $\kappa^{(\varepsilon)}$ is a normalizing constant. Moreover, if the position state of Alice collapses to $\ket{\xi_j} \in \mathcal{B}_2$ after the observable $M_2$, the total state $\ket{\varPsi_*^{(\varepsilon)}}$ is changed to the normalized state of
\begin{align}
\ket{\varPsi_*^{(j,\,\varepsilon)}}
&= \displaystyle\frac{1}{\kappa^{(j,\,\varepsilon)}}[\ket{\xi_j} \otimes \ket{\eta_\varepsilon} \nonumber \otimes \{\braket{\eta_\varepsilon|(\braket{\xi_j|2}Q_1+\braket{\xi_j|0}P_1|{\phi}}\braket{R|{C_2}|{\psi}})\ket{R}  \nonumber\\
&\hspace{40mm}+ \braket{\eta_\varepsilon|\braket{\xi_j|0}Q_1 + \braket{\xi_j|{-2}}P_1|\phi}\braket{L|{C_2}|{\psi}}\ket{L}\}] \nonumber \\
&= \ket{\xi_j} \otimes \ket{\eta_\varepsilon} \otimes \displaystyle\frac{\tilde{V}^{(j,\,\varepsilon)}}{\kappa^{(j,\,\varepsilon)}}\ket{\phi},\nonumber
\end{align}
where
\begin{align}\label{vje}
\tilde{V}^{(j,\,\varepsilon)} :=  \onebytwo {\bra{\eta_\varepsilon}(\braket{\xi_j|2}Q_1+\braket{\xi_j|0}P_1)\braket{R|{C_2}|{\psi}}}{\bra{\eta_\varepsilon}(\braket{\xi_j|0}Q_1 + \braket{\xi_j|{-2}}P_1)\braket{L|{C_2}|{\psi}}},
\end{align}
and $\kappa^{(j,\,\varepsilon)}$ is a normalizing constant.  Note that the amplitudes are inserted into the third slots in the above expression. Now, because $\|\ket{\xi_j} \otimes \ket{\eta_\varepsilon}\| = 1$,
\begin{eqnarray}
\kappa^{(j,\,\varepsilon)} = \| \ket{\xi_j}\otimes\ket{\eta_\varepsilon}\otimes\tilde{V}^{(j,\,\varepsilon)}\ket{\phi} \| = \|\tilde{V}^{(j,\,\varepsilon)}\ket{\phi} \|.\nonumber
\end{eqnarray}
Here, putting
\begin{align}
V^{(j,\,\varepsilon)} = \frac{\tilde{V}^{(j,\,\varepsilon)}}{\kappa^{(j,\,\varepsilon)}}\text{\quad and\quad} \ket{\varPhi^{(j,\,\varepsilon)}} = \frac{\tilde{V}^{(j,\,\varepsilon)}}{\kappa^{(j,\,\varepsilon)}}\ket{\phi} = V^{(j,\,\varepsilon)}\ket{\phi},\nonumber
\end{align}
we obtain the desired conclusion. 
\end{proof}
Let us put $\alpha_{jk} = \braket{j|H_1|k}\,(j,\,k\in \{0,\,\pm 2\})$ and $\beta_{\varepsilon}=\braket{\varepsilon |C_2|\psi}\,(\varepsilon \in \{L,\,R\})$. Then $\tilde{V}^{(j,\,\varepsilon)}$ is re-expressed by the following:
\begin{align}
\tilde{V}^{(j,\,\varepsilon)} =\onebytwo {\bra{\vr^{(j,\,\varepsilon)}}}{\bra{\vl^{(j,\,\varepsilon)}}} = \displaystyle\onebytwo {\bra{\eta_\varepsilon}
\twobytwo{\overline{\alpha_{2j}}\beta_R}{0}{0}{\overline{\alpha_{0j}}\beta_R}
}{\bra{\eta_\varepsilon}
\twobytwo{\overline{\alpha_{0j}}\beta_L}{0}{0}{\overline{\alpha_{(-2)j}}\beta_L} 
}C_1,  \label{v}
\end{align}
where
\begin{align}
\bra{\vr^{(j,\,\varepsilon)}}=&\bra{\eta_\varepsilon}(\braket{\xi_j|2}Q_1+\braket{\xi_j|0}P_1)\braket{R|{C_2}|{\psi}}, \label{eq:vR}\\
 {\bra{\vl^{(j,\,\varepsilon)}}} =&  {\bra{\eta_\varepsilon}(\braket{\xi_j|0}Q_1 + \braket{\xi_j|{-2}}P_1)\braket{L|{C_2}|{\psi}}},\label{eq:vL}
 \end{align}
$P_1=\ket{L}\bra{L}C_1${\rm , and }$Q_1=\ket{R}\bra{R}C_1$. We will use this expression later.

\subsection{Rewrite of the accomplishment of teleportation}
The following lemma seems to be simple, but plays an important role later.
\begin{lemma}\label{l1}
{\rm The following two statements are equivalent for $V\in M_n(\mathbb{C})$}:
\begin{enumerate}[{\rm (i)}]
\item {\rm There exists $U\in \mathrm{U}(n)$ such that for any $\phi\in \mathbb{C}^n\backslash\{0\}$, there exists a complex value $\kappa=\kappa(\phi)$ such that}
\begin{align}\nonumber
UV\phi =\kappa(\phi)\phi.
\end{align}
\item {\rm There exists a  complex number $\kappa$ such that }
\begin{align}\nonumber
V\in \kappa\mathrm{U}(n).
\end{align}
\end{enumerate}
\end{lemma}
\begin{proof}
\quad Assume (i) holds. For any $\phi\in \mathbb{C}^n$, $UV\phi = \kappa(\phi)\phi \,\,\Longleftrightarrow\,\, (UV-\kappa(\phi)I)\phi = 0 \,\,\Longleftrightarrow\,\,$ eigenvector of $UV$ is every $\phi \in \mathbb{C}^n\setminus\{0\}$. That is equivalent to $UV = \kappa(\phi)I$. Since $U$ and $V$ are independent of $\phi$, the eigenvalue $\kappa(\phi)$ must be independent of $\phi$. So (ii) holds. The converse is obvious.
\end{proof}
By using Lemma~\ref{l1}, the following lemma is completed:
%%%
\begin{lemma}\label{thm:unitarity}
\begin{multline*}
{\bfit T} \in \mathcal{T}\,\, \Longleftrightarrow\,\, 
\text{\rm for any } (j,\,\varepsilon)\in \{-2,\,0,\,2\}\times \{R,\,L\} \\ \text{\rm there exists } \kappa=\kappa^{(j,\,\varepsilon)} \text{\rm \:such that } \displaystyle\tilde{V}^{(j,\,\varepsilon)} \in \kappa {\rm U}(2).
\end{multline*}
\end{lemma}
%%%%
\begin{proof}
Let $\ket{\Phi_*^{(j,\,\varepsilon)}}\in \mathcal{H}$ be the final state after obtaining the observed values $(j,\,\varepsilon)$; that is, there exists $\ket{\Psi_*^{(j,\,\varepsilon)}}\in\Hc{B}$ such that $\ket{\Phi_*^{(j,\,\varepsilon)}}=\ket{\xi_j}\otimes\ket{\eta_\varepsilon}\otimes\ket{\Psi_*^{(j,\,\varepsilon)}}$. 
\quad By the definition of $\mathcal{T}$ and Proposition~\ref{prop:finalstate},  ${\bfit T} \in \mathcal{T}$ if and only if there must exist a unitary matrix $U^{(j,\,\varepsilon)}$ on $\Hc{B}$ such that 
\begin{align}
U^{(j,\,\varepsilon)}\ket{\varPhi^{(j,\,\varepsilon)}} =U^{(j,\,\varepsilon)}\displaystyle\frac{\tilde{V}^{(j,\,\varepsilon)}}{\kappa^{(j,\,\varepsilon)}}\ket{\phi}=  \ket{\phi}
\,\,\Longleftrightarrow\,\,
U^{(j,\,\varepsilon)}\tilde{V}^{(j,\,\varepsilon)}\ket{\phi} = \kappa^{(j,\,\varepsilon)}\ket{\phi}.\nonumber
\end{align}
Here, because $\kappa^{(j,\,\varepsilon)} =\|\tilde{V}^{(j,\,\varepsilon)}\ket{\phi}\|$, this is equivalent to the following by Lemma~\ref{l1}: $\kappa^{(j,\,\varepsilon)}$ is independent of $\ket{\phi}$ and
\begin{align}
\tilde{V}^{(j,\,\varepsilon)}\in \kappa^{(j,\,\varepsilon)}\mathrm{U}(2).\nonumber
\end{align}
\quad\vspace{-5.5\baselineskip}\\
\begin{flushright}\end{flushright}
\end{proof}
\par
In the next section, we will apply the statement of Lemma~\ref{thm:unitarity} and the expression of $\tilde{V}^{(j,\,\varepsilon)}$ in (\ref{v}).

\subsection{A necessary condition of measurement}
In this section, we will show that to accomplish the quantum teleportation, the eigenbasis of the observables on  $\mathcal{B}_1$ and $\mathcal{B}_2$ must be different from each computational standard basis. More precisely, we obtain the following theorem:
\begin{lemma}
{\rm If ${\bfit T} \in \mathcal{T}$, $H_1 \neq I_2$ and $H_2 \neq I_\infty$.}
\end{lemma}
\begin{proof}
\quad We show the contrapositive of the theorem: if $H_1 = I_{2}$ or $H_2 = I_{\infty}$, ${\bfit T} \notin \mathcal{T}$, that is, by Lemma~\ref{thm:unitarity} and (\ref{v}),
\begin{eqnarray}
\onebytwo {\bra{\vr^{(j,\,\varepsilon)}}}{\bra{\vl^{(j,\,\varepsilon)}}} = \displaystyle\onebytwo {\bra{\eta_\varepsilon}
\twobytwo{\overline{\alpha_{2j}}\beta_R}{0}{0}{\overline{\alpha_{0j}}\beta_R}
}{\bra{\eta_\varepsilon}
\twobytwo{\overline{\alpha_{0j}}\beta_L}{0}{0}{\overline{\alpha_{{(-2)}j}}\beta_L} 
}C_1 \notin {}^\forall\kappa{\rm U}(2).\,\,
 \label{imp} 
\end{eqnarray} \par
In case of $H_1 = I_{2}$, $\ket{\eta_\varepsilon}$ is equal to $\ket{\varepsilon}$, so
\begin{eqnarray}
\onebytwo {\bra{\vr^{(j,\,\varepsilon)}}}{\bra{\vl^{(j,\,\varepsilon)}}} = \displaystyle\onebytwo {\bra{\varepsilon}
\twobytwo{\overline{\alpha_{2j}}\beta_R}{0}{0}{\overline{\alpha_{0j}}\beta_R}
}{\bra{\varepsilon}
\twobytwo{\overline{\alpha_{0j}}\beta_L}{0}{0}{\overline{\alpha_{{(-2)}j}}\beta_L} 
}C_1. \nonumber
\end{eqnarray}\par
Now, when $(j,\,\varepsilon) = (j,\,R)$, we obtain
\begin{eqnarray}
\onebytwo { [1\,\,\,\,0]
\twobytwo{\overline{\alpha_{2j}}\beta_R}{0}{0}{\overline{\alpha_{0j}}\beta_R}
}{ \left[ 1\,\,\,\,0 \right]
\twobytwo{\overline{\alpha_{0j}}\beta_L}{0}{0}{\overline{\alpha_{{(-2)}j}}\beta_L} 
}=
\twobytwo{\overline{\alpha_{2j}}\beta_R}{0}{\overline{\alpha_{0j}}\beta_L}{0}.\nonumber
\end{eqnarray}
It is followed by ${\rm det}\onebytwo {\bra{\vr^{(j,\,R)}}}{\bra{\vl^{(j,\,R)}}}= 0$, and it implies (\ref{imp}).\par
In case of $H_2 = I_{\infty}$, $\ket{\xi_j}$ is equal to $\ket{j}$, so
\begin{eqnarray}\nonumber
\onebytwo {\bra{\vr^{(j,\,\varepsilon)}}}{\bra{\vl^{(j,\,\varepsilon)}}}= \displaystyle\onebytwo {
\bra{\eta_\varepsilon}
\twobytwo{\overline{\alpha_{2j}}\beta_R}{0}{0}{\overline{\alpha_{0j}}\beta_R}
}{\bra{\eta_\varepsilon}
\twobytwo{\overline{\alpha_{0j}}\beta_L}{0}{0}{\overline{\alpha_{{(-2)}j}}\beta_L} 
}C_1
= \displaystyle\onebytwo {
\bra{\eta_\varepsilon}
\twobytwo{\delta_{2j}\beta_R}{0}{0}{\delta_{0j}\beta_R}
}{\bra{\eta_\varepsilon}
\twobytwo{\delta_{0j}\beta_L}{0}{0}{\delta_{(-2)j}\beta_L} 
}C_1,
\end{eqnarray}
where
\begin{align}\nonumber
\braket{\xi_j|k}=\braket{j|k}=\delta_{jk}=
\left\{ \begin{array}{ll}
1 & (j=k)\\
0 & (j\neq k)
\end{array} \right. .
\end{align}
\par
Now, we put $H_1 = \twobytwo{a}{b}{c}{d}$. Because $\ket{\eta_\varepsilon} = H_1 \ket{\varepsilon}$, we can rewrite $\onebytwo {\bra{\vr^{(j,\,\varepsilon)}}}{\bra{\vl^{(j,\,\varepsilon)}}}$ as following:
\begin{eqnarray}\nonumber
\onebytwo {\bra{\vr^{(j,\,\varepsilon)}}}{\bra{\vl^{(j,\,\varepsilon)}}} = 
\onebytwo {
\bra{\varepsilon}H_1^\dagger
\twobytwo{\delta_{2j}\beta_R}{0}{0}{\delta_{0j}\beta_R}
}{\bra{\varepsilon}H_1^\dagger
\twobytwo{\delta_{0j}\beta_L}{0}{0}{\delta_{(-2)j}\beta_L} 
}C_1  = 
\onebytwo {
\bra{\varepsilon}\twobytwo{\overline{a}\delta_{2j}\beta_R}{\overline{c}\delta_{0j}\beta_R}{\overline{b}\delta_{2j}\beta_R}{\overline{d}\delta_{0j}\beta_R}
}{\bra{\varepsilon}\twobytwo{\overline{a}\delta_{0j}\beta_L}{\overline{c}\delta_{(-2)j}\beta_L}{\overline{b}\delta_{0j}\beta_L}{\overline{d}\delta_{(-2)j}\beta_L}}C_1. \nonumber
\end{eqnarray}\par
Under here, if $(j,\,\varepsilon) = (2,\,R)$,
\begin{eqnarray}
\onebytwo {
[1\,\,\,\,0]\twobytwo{\overline{a}\beta_R}{0}{\overline{b}\beta_R}{0}
}{\left[1\,\,\,\,0\right]\twobytwo{0}{0}{0}{0}} \nonumber
=\twobytwo{\overline{a}\beta_R}{0}{0}{0}.
\end{eqnarray}
It is followed by ${\rm det}\onebytwo {\bra{\vr^{(2,\,R)}}}{\bra{\vl^{(2,\,R)}}}= 0$, and it implies  (\ref{imp}). 
\end{proof}
\subsection{Two conditions for $\bra{v_{L}^{(j,\,\varepsilon)}}$, $\bra{v_{R}^{(j,\,\varepsilon)}}$}
By Lemma~\ref{thm:unitarity}, the problem is reduced to find a condition for the unitarity of $\tilde{V}^{(j,\,\varepsilon)}$ except a constant multiplicity. Since 
    \[ \tilde{V}^{(j,\,\varepsilon)}=\begin{bmatrix} \bra{v_R^{(j,\,\varepsilon)}} \\ \bra{v_L^{(j,\,\varepsilon)}} \end{bmatrix}, \]
the two vectors in $\Hc{B}$ must satisfy the following two conditions as the corollary of Lemma~\ref{thm:unitarity}.
\begin{corollary}
{\rm $\bfit T\in \mathcal{T}$ if and only if the two row vectors of $\tilde{V}^{(j,\,\varepsilon)};$  $\,\,\bra{\vr^{(j,\,\varepsilon)}}$ and $\bra{\vl^{(j,\,\varepsilon)}}$, satisfy}
\begin{eqnarray}
[{\bf Condition\,I}]: &\|\vr^{(j,\,\varepsilon)}\|^2=\|\vl^{(j,\,\varepsilon)}\|^2\nonumber \\
\left[{\bf Condition\,II}\right]: & \braket{\vr^{(j,\,\varepsilon)}|\vl^{(j,\,\varepsilon)}}=0 \nonumber
\end{eqnarray}
{\rm for any observed values $(j,\,\varepsilon)$.}
\end{corollary}
\begin{proof}
By the expression of $\tilde{V}^{(j,\,\varepsilon)}$ in (\ref{v}) and Lemma~\ref{thm:unitarity}, we obtain the desired condition. 
\end{proof}
From now on, we find more useful equivalent expressions of Conditions I and II.
\subsection{Equivalent expression of $[{\bf Condition\,I}]$}
From the definition of Condition I and the expressions of $\bra{v_R^{(j,\,\varepsilon)}}$ and $\bra{v_L^{(j,\,\varepsilon)}}$ in (\ref{v}), we have
\begin{align}
{\bf [Condition\, I]} & \Leftrightarrow  ||\vr ^{(j,\,\varepsilon)}||^2 = ||\vl ^{(j,\,\varepsilon)}||^2 \nonumber \\
&\Leftrightarrow  \scalebox{0.8}[1]{$\bra{\eta_\varepsilon}
\left[ \begin{array}{cc}
|\alpha_{2j}|^2|\beta_R|^2 - |\alpha_{0j}|^2|\beta_L|^2 & 0 \\
0 & |\alpha_{0j}|^2|\beta_R|^2 - |\alpha_{(-2)j}|^2|\beta_L|^2
\end{array} \right]
\ket{\eta_\varepsilon} = 0. $}\label{3}
\end{align}
Here, we put $A:=|\alpha_{2j}|^2|\beta_R|^2 - |\alpha_{0j}|^2|\beta_L|^2$ and $B:=|\alpha_{0j}|^2|\beta_R|^2 - |\alpha_{(-2)j}|^2|\beta_L|^2$.
\begin{align}
(\ref{3}) &\Longleftrightarrow  \bra{\eta_\varepsilon}
\left[ \begin{array}{cc}
A & 0 \\
0 & B
\end{array} \right]
\ket{\eta_\varepsilon} = 0 \nonumber \\
& \Longleftrightarrow 
\left(\begin{array}{l}
X_1 : ``\left[ \begin{array}{cc}
A & 0 \\
0 & B
\end{array} \right] = O\," \\ {\rm or} \\
Y_1 : ``\twobytwo{A}{0}{0}{B}\neq O\text{ and } \left[ \begin{array}{cc}
A & 0 \\
0 & B
\end{array} \right]
\ket{\eta_\varepsilon} ={}^\exists \lambda_{j,\,\varepsilon} \ket{\eta_{\lnot\varepsilon}},"
\end{array}\right.\nonumber
\end{align}
where $\lambda_{j,\,\varepsilon}\in\mathbb{C}$.  Then, we have Condition~I $= ``X_1 \lor Y_1$ for any $(j,\,\varepsilon)$" and in the following, we will transform $X_1$ and $Y_1$, respectively.
\subsubsection{Equivalent transformation of $X_1$}
{ The condition $X_1$ can be characterized by the following more practical condition using the parameters $|\alpha_{jk}|$,  $|\beta_{R}|$ and $|\beta_{L}|$, which decide $H_2$ and $C_2$:}
\begin{lemma}
\begin{eqnarray}
&X_1\,\,\Longleftrightarrow\,\,&\displaystyle|\alpha_{jk}|=\frac{1}{\sqrt{3}} \text{\rm\; for all $j, k \in \{ 0, \pm 2\}$ }
\text{\rm\; and }|\beta_R|=|\beta_L|=\frac{1}{\sqrt{2}}.\nonumber
\end{eqnarray}
\end{lemma}
\begin{proof}
Assume $|\alpha_{jk}|=1/\sqrt{3}$ for all $k,j$  and $|\beta_R|=|\beta_L|=1/\sqrt{2}$, it is easy to check that $X_1$ holds. 
Let us consider the inverse.
Assume $X_1$ holds. In this case, we obtain
\begin{align}
A &=|\alpha_{2j}|^2|\beta_R|^2 - |\alpha_{0j}|^2|\beta_L|^2=0, \nonumber \\
B &=|\alpha_{0j}|^2|\beta_R|^2 - |\alpha_{(-2)j}|^2|\beta_L|^2=0, \nonumber
\end{align}
that is, 
\begin{align}
\left[ \begin{array}{cc}
|\alpha_{2j}|^2 & -|\alpha_{0j}|^2 \\
|\alpha_{0j}|^2 & -|\alpha_{(-2)j}|^2 \\
\end{array} \right]
\left[ \begin{array}{c}
|\beta_R|^2 \\
|\beta_L|^2
\end{array} \right]
=
{\bf 0}. \nonumber
\end{align}
Because of $\T[|\beta_R|^2\,\,|\beta_L|^2]\neq {\bf 0},$ we have
\begin{eqnarray}
{\rm det} \begin{bmatrix}
|\alpha_{2j}|^2 & -|\alpha_{0j}|^2 \\
|\alpha_{0j}|^2 & -|\alpha_{(-2)j}|^2 \\
\end{bmatrix} 
=0\nonumber
\end{eqnarray}
This is equivalent to 
\begin{equation}|\alpha_{2j}|^2 |\alpha_{(-2)j}|^2=\left(|\alpha_{0j}|^2\right)^2. \label{prod} 
\end{equation}
On the other hand, by the unitarity of $\tilde{H}_2$, we have
\begin{equation}
|\alpha_{2j}|^2 + |\alpha_{(-2)j}|^2 =1 - |\alpha_{0j}|^2. \label{sum}
\end{equation}
for any $j=-2,0,2$. 
By (\ref{prod}) and (\ref{sum}), $|\alpha_{2j}|^2$, $|\alpha_{(-2)j}|^2$ are the solutions of the following quadratic equation:
\begin{eqnarray}
t^2 - (1 - |\alpha_{0j}|^2)t + \left(|\alpha_{0j}|^2\right)^2 = 0. \nonumber
\end{eqnarray}
Its solution is
\begin{eqnarray}
t = \frac{1 - |\alpha_{0j}|^2 \pm \sqrt{D}}{2},\quad D = -(3|\alpha_{0j}|^2-1)(|\alpha_{0j}|^2+1). \nonumber
\end{eqnarray}
Here, because the solution $t$ is a real number, the discriminant  $D\geq 0$, i.e., $3|\alpha_{0j}|^2-1\leq 0$. Therefore, because $|\alpha_{0j}|\geq 0$,
\begin{eqnarray}
0 \leq |\alpha_{0j}|^2 \leq \frac{1}{3}. \nonumber
\end{eqnarray}\par
Here, the necessary condition for the unitarity of $\tilde{H}_2$ that $|\alpha_{02}|^2 + |\alpha_{00}|^2 + |\alpha_{0{(-2)}}|^2=1$ is satisfied by only the case for
\begin{eqnarray}
|\alpha_{02}|^2=|\alpha_{00}|^2=|\alpha_{0{(-2)}}|^2=\frac{1}{3}.\nonumber
\end{eqnarray}
Hence, for $j\in \{0,\,\pm 2\}$, we obtain $D=0$, and then $t=1/3$ holds. Therefore, for $j,\,k\in \{0,\,\pm 2\}$,
\begin{eqnarray}
|\alpha_{jk}|=\frac{1}{\sqrt{3}},\nonumber
\end{eqnarray}
which implies,
\begin{align}\nonumber
A=B=\frac{1}{3}(|\beta_R|^2-|\beta_L|^2)=0\,\,\Longleftrightarrow\,\, |\beta_R|=|\beta_L| =\frac{1}{\sqrt{2}}.
\end{align}
\quad\vspace{-4.5\baselineskip}\\
\begin{flushright}\end{flushright}
\end{proof}
%%%%%
\subsubsection{Equivalent transformation of $Y_1$}
{ The condition $Y_1$ is equivalently deformed by the following lemma. This shows that the parameters of $H_1$ are independent of the others.}
\begin{lemma}{\rm Let measurement operator of $\Hc{A}$ be
\begin{eqnarray}
H_1 =\left[ \begin{array}{cc}
a & b \\
c & d
\end{array} \right],\nonumber
\end{eqnarray}
which is unitary. Then, we have}
\begin{align}
Y_1\,\,\Longleftrightarrow\,\,&
|\alpha_{2j}|^2|\beta_R|^2 - |\alpha_{0j}|^2|\beta_L|^2 =-|\alpha_{0j}|^2|\beta_R|^2 + |\alpha_{(-2)j}|^2|\beta_L|^2\nonumber\\
&\text{\rm\; for all $ j, k \in \{ 0, \pm 2\}$}\text{\rm\; and}\,\, |a|=|b|. \nonumber 
\end{align}
\end{lemma}
\begin{proof}
First let us consider the proof of the ``$\Leftarrow$" direction. 
It holds
    \begin{align}
        \begin{bmatrix}1 & 0 \\ 0 & -1 \end{bmatrix} \ket{\eta_R} &= \begin{bmatrix} b \\ -d \end{bmatrix}
        =\begin{bmatrix} (b/a)\cdot a \\ (\bar{a}/\bar{b})\cdot c \end{bmatrix} \notag \\
        &= \frac{\,b\,}{a} \begin{bmatrix} b \\ d \end{bmatrix} = \frac{\,a\,}{b} \ket{\eta_L} \label{eq:etaR},
    \end{align}
{where $\ket{\eta_\varepsilon} = H_1 \ket{\varepsilon}$.} Here, the second equality derives from  $c=-\varDelta\bar{b}$ and $d=\varDelta\bar{a}$, where $\varDelta=\det(H_1)$ by the unitarity of $H_1$ and the third equality comes from the last assumption of $|a|=|b|$. 
In the same way, we obtain
    \begin{align}\label{eq:etaL}
        \begin{bmatrix}1 & 0 \\ 0 & -1 \end{bmatrix} \ket{\eta_L} &=  \frac{\,a\,}{b} \ket{\eta_R}.
    \end{align}
The first assumption implies $A=-B$. Then, (\ref{eq:etaR}) and (\ref{eq:etaL}) include
    \[ \begin{bmatrix} A & 0 \\ 0 & B \end{bmatrix} \ket{\eta_R}= A\cdot\frac{\,b\,}{a}\ket{\eta_L}
    \text{ and }
    \begin{bmatrix} A & 0 \\ 0 & B \end{bmatrix} \ket{\eta_L}= A\cdot\frac{\,a\,}{b}\ket{\eta_R}
    \]
Thus, the condition $Y_1$ holds. 
Secondly, assume $Y_1$ holds.
In this case, there exist $\lambda$ and $\lambda'$ such that
\begin{eqnarray}
\left[ \begin{array}{cc}
A & 0 \\
0 & B
\end{array} \right]
\ket{\eta_R} = \lambda \ket{\eta_{L}}\,\,{\rm and}\,\,
\left[ \begin{array}{cc}
A & 0 \\
0 & B
\end{array} \right]
\ket{\eta_L} = \lambda' \ket{\eta_{R}}.
\label{rl}
\end{eqnarray}
Therefore, 
\begin{subequations}
    \begin{empheq}[left = {(\ref{rl})\,\,\Longleftrightarrow\,\,\empheqlbrace \,}, right = {}]{align}
& \left[ \begin{array}{cc}
A & 0 \\
0 & B
\end{array} \right]
H_1 \ket{R} = \lambda H_1 \ket{L}\,  \label{a} \\
    & {\rm and} \nonumber\\
    & \left[ \begin{array}{cc}
A & 0 \\
0 & B
\end{array} \right]
H_1 \ket{L} =\lambda' H_1 \ket{R}. \label{b}
    \end{empheq}
\end{subequations}\par
Let us give further transformation  of (\ref{a}). Because $H_1$ is unitary,
\begin{align}
& H_1^\dagger \twobytwo{A}{0}{0}{B} H_1 \ket{R} = \lambda \ket{L} \nonumber \\
\Longleftrightarrow & \onebytwo{|a|^2A+|c|^2B}{a\overline{b}A+c\overline{d}B} = \onebytwo{0}{\lambda}. \label{c}
\end{align}\par
Similarly, (\ref{b}) is equivalently deformed as follows:
\begin{align}
& H_1^\dagger \twobytwo{A}{0}{0}{B} H_1 \ket{L} = \lambda' \ket{R} \nonumber \\
\Longleftrightarrow& \onebytwo{\overline{a}bA+\overline{c}dB}{|b|^2A+|d|^2B} = \onebytwo{\lambda'}{0}. \label{d}
\end{align} 
Therefore, (\ref{rl}) is equivalent to (\ref{c}) and (\ref{d}), and these are also equivalent to
\begin{eqnarray}
\onebytwo{|a|^2A+|c|^2B}{|b|^2A+|d|^2B}=\twobytwo{|a|^2}{1-|a|^2}{1-|a|^2}{|a|^2}\onebytwo{A}{B} = {\bf 0} \label{e} 
\end{eqnarray}
\begin{center}
and
\end{center}
\begin{eqnarray}
\onebytwo{a\overline{b}A+c\overline{d}B}{\overline{a}bA+\overline{c}dB}=\twobytwo{a\overline{b}}{c\overline{d}}{\overline{a}b}{\overline{c}d}\onebytwo{A}{B} = \twovec{\lambda}{\lambda'}
\label{f}
\end{eqnarray}
Here, we used in (\ref{e}), the unitarity of $H_1$, $|a|^2=|d|^2=1-|b|^2=1-|c|^2$.
Moreover, because of the assumption $\T[A,\,B] \neq {\bf 0}$,
\begin{eqnarray}
\det\twobytwo{|a|^2}{1-|a|^2}{1-|a|^2}{|a|^2}=0 \,\,\Longleftrightarrow\,\, |a| = \frac{1}{\sqrt{2}}. \nonumber
\end{eqnarray}
Then, we have $|a|=|b|$. 
By substituting this result to (\ref{e}), we obtain
\begin{eqnarray}\label{eq:A+B=0}
A+B=0, %\,\,\Longleftrightarrow\,\, A=-B.
\end{eqnarray}
which is equivalent to 
\[ |\alpha_{2j}|^2|\beta_R|^2 - |\alpha_{0j}|^2|\beta_L|^2 =-|\alpha_{0j}|^2|\beta_R|^2 + |\alpha_{(-2)j}|^2|\beta_L|^2 \]
for all $j$.

\end{proof}
Note that, by substituting (\ref{eq:A+B=0}) to (\ref{f}), we obtain
\begin{align}
a\bar{b}-c\bar{d}=\frac{\lambda'}{A},\; \bar{a}b-\bar{c}d=\frac{\lambda}{A}. \nonumber
\end{align}
The unitarity of $H_1$ implies $d=\varDelta \bar{a}$, $c=-\varDelta \bar{b}$, where $\varDelta=\det(H_1)$. 
Therefore, we obtain
the constants of the Condition $Y_1$ are  \[\lambda'=2a\bar{b}\cdot A=\frac{\,b\,}{a}\cdot A \text{ and } \lambda=2\bar{a}b\cdot A=\frac{\,a\,}{b}\cdot A\] 
since  $|a|=|b|=1/\sqrt{2}$.
%%%%%%%%%%%%%%%%%%%%%%
\subsection{Calculation of $\left[{\bf Condition\,II}\right]$}
From the definition of [Condition I] and the expressions of $\bra{v_R^{(j,\,\varepsilon)}}$ and $\bra{v_L^{(j,\,\varepsilon)}}$ in (\ref{v}), we have
\begin{align}
{\bf [Condition\,II]}  \Longleftrightarrow & \braket{\vr ^{(j,\,\varepsilon)}|\vl ^{(j,\,\varepsilon)}} = 0 \nonumber \\
\Longleftrightarrow & \bra{\eta_\varepsilon}
\left[ \begin{array}{cc}
\beta_R \alpha_{2j} \overline{\alpha_{0j}\beta_L} & 0 \\
0 & \beta_R \alpha_{0j} \overline{\alpha_{({-2})j}\beta_L}
\end{array} \right]
\ket{\eta_\varepsilon} = 0 \label{3'}
\end{align}
Putting $A':=\beta_R \alpha_{2j} \overline{\alpha_{0j}\beta_L}$ and $B':=\beta_R \alpha_{0j} \overline{\alpha_{({-2})j}\beta_L}$, we decompose (\ref{3'}) into the conditions $X_2$ and $Y_2$, as follows.
\begin{align}
(\ref{3'}) \,\,\Longleftrightarrow\,\, & \bra{\eta_\varepsilon}
\left[ \begin{array}{cc}
A' & 0 \\
0 & B'
\end{array} \right]
\ket{\eta_\varepsilon} = 0 \nonumber \\
\Longleftrightarrow\,\,& 
\left( \begin{array}{l}
X_2 : ``\left[ \begin{array}{cc}
A' & 0 \\
0 & B'
\end{array} \right] = O\," \\ {\rm or} \\
Y_2 : ``\twobytwo{A'}{0}{0}{B'}\neq O\text{ and }\left[ \begin{array}{cc}
A' & 0 \\
0 & B'
\end{array} \right]
\ket{\eta_\varepsilon} ={}^\exists\mu_{j,\,\varepsilon} \ket{\eta_{\lnot\varepsilon}},"
\end{array}\right. \nonumber
\end{align}
where $\mu_{j,\,\varepsilon}\in\mathbb{C}$.  
Then we obtain [Condition II]$=X_2 \lor Y_2$. 
We will transform $X_2$ and $Y_2$ to more useful forms.
\subsubsection{Equivalent transformation of $X_2$}
{The condition $X_2$ is characterized only by the parameters of $H_2$ as follows:}
\begin{lemma}
{\rm {Let $\bfit{H}$ be the set of three dimensional unitary matrices defined by 
    \begin{equation}\label{eq:H} \bfit{H}=\left\{  
    \begin{bmatrix} p & r & 0 \\  0 & 0 & t  \\q & s & 0 \end{bmatrix},\;
    \begin{bmatrix} p & 0 & r \\ 0 & t & 0 \\ q & 0 & s  \end{bmatrix},\;
    \begin{bmatrix} 0 & p & r \\ t & 0 & 0 \\ 0 & q & s \end{bmatrix}
    \in {\rm U}(3) \;:\; 
    |p|=|q|   \right\} 
    \end{equation}
}The condition $X_2$ is equivalent to the following condition; 
\[{ H_2=\tilde{H}_2 \oplus I_{\infty}\text{\,\, with\,\,\,} \tilde{H}_2=
\left[ \begin{array}{ccc}
\alpha_{22} & \alpha_{20} & \alpha_{2{(-2)}} \\
\alpha_{02} & \alpha_{00} & \alpha_{0{(-2)}} \\
\alpha_{{(-2)}2} & \alpha_{{(-2)}0} & \alpha_{{(-2)}{(-2)}} 
\end{array} \right]\in \bfit{H}. }\]}
\end{lemma}
\begin{proof}
Assume $\tilde{H}_2\in \bfit{H}$. Then, each raw vector of $\tilde{H}_2$ is of the form $[*,\;0\;,*]$ or $[0,\;*,\;0]$, where ``$*$" takes a nonzero value. Since the computational basis of $\tilde{H}_2$ is $\ket{-2},\ket{0},\ket{2}$ by this order, it holds that $\alpha_{2j}\alpha_{0j}=\alpha_{(-2)j}\alpha_{0j}=0$ for any $j\in\{-2,0,2\}$. Then, we have $A'=B'=0$ which implies the condition $X_2$.
On the other hand, assume the condition $X_2$. In this case, for $A'$ and $B'$, the followings are held:
\begin{align}
A' &=\beta_R \alpha_{2j} \overline{\alpha_{0j}\beta_L}=0, \nonumber \\
B' &=\beta_R \alpha_{0j} \overline{\alpha_{({-2})j}\beta_L}=0. \nonumber
\end{align}
Therefore, 
\begin{align}
& |\beta_R \alpha_{2j}\alpha_{0j}\beta_L| = |\beta_R \alpha_{0j}\alpha_{j({-2})}\beta_L| = 0 \nonumber \\
 \Longleftrightarrow\,\,& |\beta_R\alpha_{0j}\beta_L| = 0\,\,\,{\rm or}\,\,\,|\alpha_{2j}| = |\alpha_{j({-2})}| = 0 \nonumber \\
 \Longleftrightarrow\,\,& 
``(\;|\beta_R|,\;|\beta_L|\;)\in \{(0,1),(1,0)\}"\notag \\ & \qquad\qquad\text{ or }``(\;|\alpha_{0j}|^2,\;|\alpha_{2j}|^2+|\alpha_{(-2)j}|^2\;)\in \{(0,1),(1,0)\}"\nonumber
\end{align}
Here, we used $|\alpha_{0j}|^2+|\alpha_{2j}|^2+|\alpha_{(-2)j}|^2=1$  due to the unitarity of $\tilde{H}_2$ in the last equivalence.
When $(\;|\beta_R|,\,|\beta_L|\;) = (0,\,1)$ or $(1,\,0)$, the determinant of $\tilde{V}^{(j,\,\varepsilon)}$ is  $\det(\tilde{V}^{(j,\,\varepsilon)})=0$ by (\ref{v}), and because of it, the matrix $\tilde{V}^{(j,\,\varepsilon)}$
does not satisfy the condition of Theorem~2. Hence, the conditions we should only impose are
\begin{align*}
({\rm a})\,\,
(\;|\alpha_{0j}|^2,|\alpha_{2j}|^2+|\alpha_{(-2)j}|^2\;) &= (0,1)\\ 
{\rm or}\quad\quad\quad\quad\quad\\
({\rm b})\,\,
(\;|\alpha_{0j}|^2,|\alpha_{2j}|^2+|\alpha_{(-2)j}|^2\;)&=(1,0)
\end{align*}
to each column vector of $\tilde{H}_2$ ($j=-2,0,2$). Each column vector satisfies the condition (a) or (b), however by the unitarity of $\tilde{H}_2$, we notice that one of the column vectors in $\tilde{H}_2$ satisfies the condition (b) and all the rest of the two column vectors satisfy (a) because every {\it raw} vector of $\tilde{H}_2$ must be a unit vector. This implies that $H_2=\tilde{H}_2\oplus I_\infty$ with $\tilde{H}_2\in \bfit{H}$.
Then, we obtained the desired conclusion. 
\end{proof}
\subsubsection{Equivalent transformation of $Y_2$}
By {a similar} discussion {to} that of the condition $Y_1$, we obtain the following lemma. {It is important that the lemma is free from constraints of Alice's coin operator $C_2$. In spite of a similar fashion of the proof, this gives us a different observation from the observation of $Y_1$.}
\begin{lemma}{\rm For all $ j, k \in \{ 0, \pm 2\}$,}
\begin{eqnarray}\nonumber
Y_2 \,\,\Longleftrightarrow\,\, \alpha_{2j} \overline{\alpha_{0j}} = -\alpha_{0j} \overline{\alpha_{({-2})j}} \,\,{\rm and}\,\,  |a|=|b|.
\end{eqnarray}
\end{lemma}

\subsection{Fusion of the conditions}
We have shown that a necessary and sufficient condition for $\bfit{T}\in \mathcal{T}$ is $(X_1 \lor Y_1) \land (X_2 \lor Y_2)$ and we have converted $X_j$ and $Y_j$ $(j=1,2)$ to useful expressions in the above discussions. 
Expanding 
\begin{equation}
(X_1 \lor Y_1) \land (X_2 \lor Y_2) = (X_1 \land X_2) \lor (X_1 \land Y_2) \lor (Y_1 \land X_2) \lor (Y_1 \land Y_2),
\nonumber
\end{equation}
we consider each case as follows to finish the proof of  Theorem~\ref{thm:main}.

\begin{center}
\begin{tabular}{|c||c|c|}\hline
& $\begin{array}{c} X_2 \\ \tilde{H_2} \in \bfit{H} \end{array}$
& $\begin{array}{c} Y_2 \\ \scalebox{0.8}[1]{$\alpha_{2j} \overline{\alpha_{0j}} = -\alpha_{0j} \overline{\alpha_{({-2})j}}$} 
\\|a|=|b| \end{array}$ \\ \hline\hline
$\begin{array}{c} X_1 \\ |\beta_R|=|\beta_L|=1/\sqrt{2} \\ |\alpha_{jk}| = 1/\sqrt{3} \end{array}$
& {\LARGE (A)} & {\LARGE (B)} \\ \hline
$\begin{array}{c} Y_1 \\  \scalebox{0.7}[1]{$|\alpha_{2j}|^2|\beta_R|^2 - |\alpha_{0j}|^2|\beta_L|^2  
=-|\alpha_{0j}|^2|\beta_R|^2 + |\alpha_{(-2)j}|^2|\beta_L|^2$} \\  |a|=|b| \end{array} $
& {\LARGE (C)} & {\LARGE (D)} \\ \hline
\end{tabular}

\begin{comment}
\begin{tabular}{|c|c|c|} \hline
 & \,\,\,$X_1$\,\,\, & \,\,\,$Y_1$\,\,\, \\
 & $|\beta_R|=|\beta_L|$ & $|\alpha_{2j}|^2|\beta_R|^2 - |\alpha_{0j}|^2|\beta_L|^2 $ \\
 &${}^\forall j,\,k \in \{0,\, \pm 2\},\,\,|\alpha_{jk}| = 1/\sqrt{3}$&$=-|\alpha_{0j}|^2|\beta_R|^2 + |\alpha_{(-2)j}|^2|\beta_L|^2$\\
 &&{\rm and} \quad $|a|=|b|$ \\
\hline
\,\,\,$X_2$\,\,\, &\raisebox{-10pt}[0pt][0pt]{\LARGE (A)}&\raisebox{-10pt}[0pt][0pt]{\LARGE (B)}\\ 
$\tilde{H_2} \in \bfit{H}$ & \\ \hline
\,\,\,$Y_2$\,\,\, &\raisebox{-15pt}[0pt][0pt]{\LARGE (C)}&\raisebox{-15pt}[0pt][0pt]{\LARGE (D)} \\
$\alpha_{2j} \overline{\alpha_{0j}} = -\alpha_{0j} \overline{\alpha_{({-2})j}}$&& \\
{\rm and}\,\,  $|a|=|b|$ && \\ \hline
\end{tabular} 　
\end{comment}
\end{center}

\begin{enumerate}[(A)]
\item $X_1 \land X_2$
%%%
\begin{lemma}\label{lem:A}
$X_1 \land X_2=\emptyset$
\end{lemma}
%%%
\begin{proof}
It is easy to see that $X_1$ and $X_2$ are contradictory each other. 
\end{proof}
\vspace{1\baselineskip}
\item $X_1 \land Y_2$
\begin{lemma}\label{lem:C}
{\rm The condition $X_1 \land Y_2$ coincides with {\bf (I)}, {\bf (II)} and {\bf (III)-(ii)} in the condition of Theorem~\ref{thm:main} for the case of $|(H_2)_{jk}|=1/\sqrt{3}$ for any $j,k\in\{-2,0,2\}$. }
\end{lemma}
\begin{proof}
Let us assume $X_1 \land Y_2$. 
By $X_1$, for $j,\,k \in \{ 0,\,\pm 2\}$,
\begin{eqnarray}
\alpha_{jk} = \frac{e^{i\arg\alpha_{jk}}}{\sqrt{3}}.\nonumber
\end{eqnarray}
We can rewrite $Y_2$ by using it as follows:
\begin{align}
&\frac{1}{\sqrt{3}}\cdot e^{i(\agmnt\alpha_{2j} - \agmnt\alpha_{0j})} = -\frac{1}{\sqrt{3}}\cdot e^{i(\agmnt\alpha_{0j} - \agmnt\alpha_{(-2)j})} \nonumber \\
\Longleftrightarrow\,\, & \agmnt\alpha_{2j} + \agmnt\alpha_{(-2)j} - 2\agmnt\alpha_{0j} \in (2\mathbb{Z}+1)\pi = \{ (2m + 1)\pi  | m \in \mathbb{Z} \}. \nonumber
\end{align}
Therefore, the condition $X_1 \land Y_2$ includes
\begin{align*}
& |a|=|b|;\\
& |\beta_R| = |\beta_L| = \displaystyle\frac{1}{\sqrt{2}};\\
& |\alpha_{jk}| = \displaystyle\frac{1}{\sqrt{3}} \text{ for any } j,k\in\{0,\,\pm2\};\\
& \agmnt\alpha_{2j} + \agmnt\alpha_{(-2)j} - 2\agmnt\alpha_{0j} \in (2\mathbb{Z}+1)\pi \text{ for any } j\in\{0,\,\pm2\};
\end{align*}
The reverse is also true.\vspace{\baselineskip}
\end{proof}
%%%%
\item $Y_1 \land X_2$
\begin{lemma}\label{lem:B}
{\rm The condition $Y_1 \land X_2$ coincides with {\bf (I)},{\bf (II)} and {\bf (III)-(i)} in the condition of  Theorem~\ref{thm:main}.} 
\end{lemma}
\begin{proof}
Let us assume $Y_1\land X_2$. By $Y_1$, the condition 
$|\alpha_{2j}|^2|\beta_R|^2 - |\alpha_{0j}|^2|\beta_L|^2=-|\alpha_{0j}|^2|\beta_R|^2 + |\alpha_{(-2)j}|^2|\beta_L|^2$ holds for any $j\in\{-2,0,2\}$, and by $X_2$, the condition $\tilde{H}_2\in\bfit{H}$ holds. Therefore, by the definition of $\bfit{H}$ in (\ref{eq:H}), 
we obtain
\begin{eqnarray}
|p\beta_R| = |q\beta_L|\quad {\rm and}\quad |r\beta_R| = |s\beta_L|\quad {\rm and} \quad |\beta_R| = |\beta_L|. \nonumber
\end{eqnarray}
Therefore, we can obtain $|\beta_R| = |\beta_L| = 1/\sqrt{2}$ from all of the condition and $|p|=|q|=|r|=|s|$. Hence, the condition $Y_1 \land X_2$ includes
\begin{eqnarray}
\tilde{H}_2 \in  \bfit{H}\quad {\rm and} \quad |\beta_R| = |\beta_L| = \frac{1}{\sqrt{2}} \quad {\rm and} \quad |a|=|b|.\nonumber
\end{eqnarray}
The reverse is also true.
\end{proof}
\par
By this result, there exist permutation matrices $\mathcal{U}$ and $\mathcal{V}$ such that $H_2$ can be expressed by
\begin{eqnarray}
H_2 = 
\mathcal{U}\left[
\begin{array}{@{\,}ccc|c@{\,}}
\frac{1}{\sqrt{2}}e^{i\agmnt\alpha_{j_1 k_1}}&\frac{1}{\sqrt{2}}e^{i\agmnt\alpha_{j_1 k_2}} & 0 &  \\
\frac{1}{\sqrt{2}}e^{i\agmnt\alpha_{j_2 k_1}}&\frac{1}{\sqrt{2}}e^{i\agmnt\alpha_{j_2 k_2}} & 0 & O\\
0 &0 & 1 & \\
\hline
 &O & &I
\end{array}
\right]\mathcal{V}.\nonumber
\end{eqnarray}
In particular, when $j_1=k_1= 2$, $j_2 = k_2 =-2$, $\agmnt\alpha_{22}=\agmnt\alpha_{({-2})2}=\agmnt\alpha_{2({-2})}=0$ and $\agmnt\alpha_{{(-2)}{(-2)}} = \pi$, the result meets the example in paper \cite{WSX17}.\vspace{1\baselineskip}
%%%%%%%
\item $Y_1 \land Y_2$
\begin{lemma}\label{lem:D}
{\rm $Y_1 \land Y_2$ coincides with {\bf (I)}, {\bf (II)} and {\bf (III)-(ii)} in the condition of Theorem~\ref{thm:main}.} 
\end{lemma}
\begin{proof}
Let us assume $Y_1\land Y_2$. Taking the absolute values to both sides of the condition $Y_2$, we obtain $|\alpha_{2j}|=|\alpha_{(-2)j}|$ for any $j\in \{-2,0,2\}$. Inserting this into the condition $Y_1$, we have 
    \[ (|\alpha_{2j}|^2+|\alpha_{0j}|^2)(|\beta_R|^2-|\beta_L|^2)=0. \]
Since $|\alpha_{2j}|,|\alpha_{0j}|>0$, we get  $|\beta_R|^2=|\beta_L|^2$. 
In the next, let us consider $Y_2$ with respect to the phase;
the condition $Y_2$ implies 
\[ \agmnt \alpha_{2j}-\agmnt \alpha_{0j}=(2m+1)\pi+\agmnt \alpha_{0j}
-\agmnt \alpha_{(-2)j}  \]
for any $m\in\mathbb{Z}$. 
This implies 
\[ \agmnt \alpha_{2j}- 2\agmnt \alpha_{0j}
+\agmnt \alpha_{(-2)j}\in (2\mathbb{Z}+1)\pi. \]
Therefore, $Y_1 \land Y_2$ includes
\begin{align*}
& |a|=|b|;\\
& |\beta_R| = |\beta_L| = \displaystyle\frac{1}{\sqrt{2}};\\
& |\alpha_{2j}| = |\alpha_{(-2)j}| \text{ for any } j\in\{0,\,\pm2\};\\
& \agmnt\alpha_{2j} + \agmnt\alpha_{(-2)j} - 2\agmnt\alpha_{0j} \in (2\mathbb{Z}+1)\pi \text{ for any } j\in\{0,\,\pm2\};
\end{align*}
The reverse is also true.
\end{proof} 
\vspace{\baselineskip}
\end{enumerate}
Combining all together with Lemmas~\ref{lem:A}--\ref{lem:D}, we complete the proof of Theorem~\ref{thm:main}. 
%%%%%%%%%%%%%%%%%%%%%%%%%%
%%%%%%%%%%%%%%%%%%%%%%%%%%
\section{Summary and Discussion}
In this paper, we extended the scheme of quantum teleportation by quantum walks introduced by Wang et al. \cite{WSX17}. First, we introduced the mathematical definition of the accomplishment of quantum teleportation by this extended scheme. Secondly, we showed a useful necessary and sufficient condition that the quantum teleportation is accomplished rigorously.  Our result classified the parameters of the setting for {the accomplishment of quantum teleportation}. Moreover, we demonstrated some examples of the scheme of the teleportation that is accomplished. Here, we identified the model proposed in the previous study as one of the examples and gave the new models of the teleportation. Moreover, we implied that we can simplify the teleportation in terms of theory and experiment. \par
{In terms of experiment, the example (1) in 4.2 has been realized \cite{CDBP20}. Using Theorem 1, we  covered all the patterns of teleportation scheme via quantum walks on $\mathbb{Z}$ and mathematically suggested that this model is the easiest one to implement. This expectation implies that the model is also the most reliable model from the perspective of accuracy of algorithm. \par
%One of our future work is to check it via machines such as ibmqx2 or ibmqx4.\par
Also, this mathematical structure itself can be discussed or extended. For example, the relationship between the number of possible measurement outcomes $t_1 = \#\{(j,\,\varepsilon)\}$ and that of possible revise operator $t_2 = \#\{ U^{(j,\,\varepsilon)}\}$ is interesting. In this paper, $t_1$ is restricted to 6 ($t_1 = \#(\{\pm 2,\,0\}\times\{R,\,L\})$). Moreover, $t_2=\#\{I_2,\,X,\,Z,\,ZX\}=4<t_1$ for example (1) in 4.2, and $t_2=6=t_1$ for example (2) or (3), both of which are from the case satisfying {\bf (III)-(ii)}. Here one question arises: can we structure some examples that satisfies both {\bf (III)-(ii)} and $t_2<t_1$? Structuring such models will lead us to implement simpler teleportation schemes if we can do. Possibly, one can also think that how the model would be if we extend it so that $t_1>6$. By adding $\ket{\xi_j}$ for $j\notin\{\pm 2,\,0\}$ to $\mathcal{B}_2$, we can extend the number of possible measurement outcomes to $t_1=2d+2$ with $d\geq 3$. This extension is meaningful when we run quantum walks more steps before measurement, and then the scheme of teleportation will be the one which is different from what we explained in this paper. It is interesting whether we can carry on teleportation such that Alice does not simply send information to Bob via the scheme. We would like to treat them as future work from the perspective of both mathematics and application.}
%Our future's work is to implement the scheme of teleportation. Moreover, applications of the properties of quantum walks to the scheme of teleportation is one of the interesting future's problems.


\begin{thebibliography}{}

\bibitem{ADZ93}
Aharonov, Y., Davidovich, L., Zagury, N.,
Quantum random walks,
\textit{Phys. Rev. A}, {\bf 48}, 1687, (1993).

\bibitem{ABNV01}
Ambainis, A., Bach, E., Nayak, A., Vishwanath, A., Watrous, J., One-dimensional quantum walks, \textit{Proc. of the 33rd Annual ACM Symposium on Theory of Computing}, 37-49, (2001).% etc

\bibitem{CFG02}
Childs, A. M., Farhi, E., Gutmann, S., An example of the difference between quantum and classical random walks, \textit{Quantum Inf. Process.}, {\bf 1}, 35-43, (2001).

\bibitem{K03}
Kempe, J., Quantum random walks - an introductory overview, \textit{Contemp. Phys.,} {\bf 44}, 307-327, (2003). 

\bibitem{Kn02}
Konno, N., Quantum random walks in one dimension, \textit{Quantum Inf. Process.}, {\bf 1}, 345-354, (2002).

\bibitem{VA08}
Venegas-Andraca, S. E., Quantum Walks for Computer Scientists, Morgan and Claypool, San Rafael (2008).

\bibitem{VA12}
Venegas-Andraca, S. E., Quantum walks: a comprehensive review, \textit{Quantum Inf. Process.}, {\bf 11}, 1015-1106, (2012).

\bibitem{P13}
Portugal, R., Quantum Walks and Search Algorithms, 2nd ed., Springer, New York (2018).

\bibitem{KI19}%%%English title?%%%
Konno, N., Ide, Y., Developments of Quantum Walks (in Japanese), Baifukan, Tokyo, (2019).

\bibitem{C09}
Childs, A. M., Universal computation by quantum walk,
\textit{Phys. Rev. Lett.}, {\bf 102}, 4586, (2013).

\bibitem{LCETK10}
Lovett, N. B., Cooper, S., Everitt, M., Trevers, M., Kendon, V.,
Universal quantum computation using the discrete-time quantum walk,
\textit{Phys. Rev. A}, {\bf 81}, 892, (2010).

\bibitem{CGW13}
Childs, A. M., Gosset, D., Webb, Z.,
Universal computation by multiparticle quantum walk,
\textit{Science}, 339, 791, (2013).

\bibitem{KRBD10}
Kitagawa, T., Rudner, M. S., Berg, E., Demler, E., Exploring topological phases with quantum walks, \textit{Phys. Rev. A}, {\bf 82}, 899, (2010).

\bibitem{HXY14}
Hong, Z. H., Xiu, L. T., Yang, H.,
A general method of selecting quantum channel for bidirectional quantum teleportation, 
\textit{Int. J. Theor. Phys.}, {\bf 53}, 1840-1847, (2014).

\bibitem{BBCJPW93}
Bennett, C. H., Brassard, G., Cr\'{e}peau, C., Jozsa, R., Peres, A., Wootters, W. K.,
Teleporting an unknown quantum state via dual classical and Einstein-Podolsky-Rosen channels, \textit{Phys. Rev. Lett.}, {\bf 70}, 1895-1899, (1993).

\bibitem{TMFLF13}
Takeda, S., Mizuta, T., Fuwa, M., Loock, P. van, Furusawa, A.,
Deterministic quantum teleportation of photonic quantum bits by a hybrid technique,
\textit{Nature}, {\bf 500}, 315-318, (2013).

\bibitem{WSX17} 
Wang, Y.,  Shang, Y., Xue, P.,
Generalized teleportation by quantum walks,
\textit{Quantum Inf. Process.},
 {\bf 16}, 221, 
(2017).
{
\bibitem{SWLR19}
Shang, Y., Wang, Y., Li, M., Lu, R.,
Quantum communication protocols by quantum walks with two coins,
\textit{EPL}, {\bf 124}, 60009, (2019).

\bibitem{LCWHL19}
Li, H. J., Chen, X. B., Wang, Y. L., Hou, Y. Y., Li, J.,
A new kind of flexible quantum teleportation of an arbitrary multi-qubit state by multi-walker quantum walks,
\textit{Quantum Inf. Process.}, {\bf 18}, 1-16, (2019).
}
\bibitem{ZYLZ20}
Zheng, T., Chang, Y., Yan, L., Zhang, S. B.,
Semi-quantum proxy signature scheme with quantum walk-based teleportation,
\textit{Int. J. Theor. Phys.}, {\bf 59}, 3145-3155, (2020).
{
\bibitem{CDBP20}
Chatterjee, Y., Devrari, V., Behera, B. K., Panigrahi, P.  K., Experimental realization of quantum teleportation using coined quantum walks, \textit{Quantum Inf. Process.}, {\bf 19}, 1-14, (2020).}

\end{thebibliography}
\end{document}